\documentclass[copyright]{eptcs}
\providecommand{\event}{J.~Dassow, G.~Pighizzini, B.~Truthe (Eds.): 11th International Workshop\\
on Descriptional Complexity of Formal Systems (DCFS 2009)} 
\providecommand{\volume}{3}
\providecommand{\anno}{2009}
\providecommand{\firstpage}{65} 
\providecommand{\eid}{6}
\usepackage{breakurl}        

\input{dcfs.tex}

\newcommand{\bdisplay}{\begin{description}\footnotesize\item[]}
\newcommand{\edisplay}{\end{description}}

\newcommand{\bquot}[1]{\begin{quotation}\small\noindent
  \textbf{#1}\hspace{\labelsep}\ignorespaces}
\newcommand{\equot}{\unskip\end{quotation}}

\def\card#1{\vert #1\vert}

\begin{document}%
     %
\title{Serializing the Parallelism in Parallel Communicating Pushdown Automata Systems}
\def\titlerunning{Serializing the Parallelism in Parallel Communicating Pushdown Automata Systems}
\author{M. Sakthi Balan
\institute{ECOM Research Lab --
  Education \& Research\\
  Infosys Technologies Limited --
  Bangalore -- 560100 -- India}
\email{sakthi\_muthiah@infosys.com}
}
\def\authorrunning{M.\,S.~Balan}
\maketitle

\newif\ifFinal
\Finaltrue
\ifFinal
\long\def\reminder#1\endreminder{\relax}\else
\long\def\reminder#1\endreminder{\par\medskip
\noindent {\bf REMINDER:} #1\par\medskip}\fi

\newcommand{\PCPA}{\mathit{PCPA}}
\newcommand{\RCPCPA}{\mathit{RCPCPA}}
\newcommand{\RCPCPAs}{\mathit{RCPCPAs}}
\newcommand{\RPCPA}{\mathit{RPCPA}}
\newcommand{\simpleRPCPA}{\mathit{simple\mbox{-}RPCPA}}
\newcommand{\IN}{\mathit{IN}}
\newcommand{\OUT}{\mathit{OUT}}
\newcommand{\dollar}{\$}
%

\begin{abstract} ~
We consider parallel communicating pushdown automata systems ($\PCPA$)
and define a property called {\em known communication} for it. We use
this property to prove that the power of a variant of $\PCPA$, called
returning centralized parallel communicating pushdown automata
($\RCPCPA$), is equivalent to that of multi-head pushdown automata. The
above result presents a new sub-class of returning parallel
communicating pushdown automata systems ($\RPCPA$) called
$\simpleRPCPA$ and we show that it can be written as a finite
intersection of multi-head pushdown automata systems.  
\end{abstract}

\section{Introduction}

Parallel communicating pushdown automata systems with communication
via stacks, studied by E.~Csu\-haj-Varj\'{u} et al~\cite{3}, is a collection
of pushdown automata working in parallel to accept a language. These
pushdown automata communicate among themselves through their stacks by
request. The component in need of some information from another
component introduces an appropriate query symbol in its stack and this
results in the whole stack information of the requested component
being transferred to the stack of the requesting component. The
parallel communicating automata systems is the counter part of the
parallel communicating grammar systems \cite{9,10}.

Similar to the parallel communicating version there is a cooperating
distributed version of collection of pushdown automata called as
cooperating distributed automata systems studied in \cite{4,5,6}
(called as multi-stack pushdown automata in \cite{4}). This is a
collection of pushdown automata that work in a sequential way
according to some protocols. The strategies considered are similar to
those defined for cooperating distributed grammar systems
\cite{1}. The power of this model with respect to various protocols
has been proved to be equivalent to that of a Turing machine.

There is a similar definition for finite state automata called as
parallel communicating finite automata systems that are finite
collections of finite state automata working independently but
communicating their states to each other by request. The notion was
introduced in \cite{8} by C. Martin-Vide et al. In this paper we
concentrate only on parallel communicating pushdown automata systems.

There are four variants of parallel communicating pushdown automata
systems defined in \cite{3},\linebreak namely,
\begin{enumerate}
\item non-centralized non-returning PCPA, denoted as $\PCPA$,
\item non-centralized returning PCPA, denoted by $\RPCPA$,
\item centralized non-returning PCPA, denoted by $CPCPA$,
\item centralized returning PCPA, denoted by $\RCPCPA$.
\end{enumerate}
In \cite{3}, it has been shown that both the non-centralized variants
are universally complete by simulating a two stack machine. Moreover,
the centralized versions of PCPA were shown to have at least the power
of a multi-head and multi-stack pushdown automata and the exact power
was left as an open problem. In~\cite{13}, it was shown that the centralized
non-returning PCPA is also universally complete. The acceptance power
of the last variant mentioned in the list -- centralized returning
PCPA was the only variant that was left open. First, it was felt that
since it has the flexibility of having more than one stacks it could
be again computationally complete. In \cite{13}, it was conjectured
that reversing the stack in the case of $\RCPCPA$ is impossible. In
this paper we present an interesting result that $\RCPCPA$ can be
simulated by a multi-head pushdown automata. This result together with
the result of \cite{3} that $\RCPCPA$ has at least the acceptance power
of multi-head pushdown automata shows that they are both equivalent
with respect to their acceptance power. This result is interesting in
the sense that it is the only variant that is not equivalent to Turing
machine when compared to other variants in parallel communicating
pushdown automata and its sequential counterpart -- distributed pushdown
automata systems \cite{6}.

In this paper, first we define a property called {\em known
  communication property} for $\PCPA$. In $\PCPA$ systems, the
communication takes place via stacks of the components. As mentioned
earlier, the querying (requesting) component places a designated query
symbol corresponding to the component from which it seeks
communication at the top of its stack and the component that was
requested transfers the contents of the stack to the requesting
component. At the state level of the communicating component it does
not know when the communication occurs and to which component the
communication occurs (if they are non-centralized\footnote{For
  centralized systems, trivially, it is the master component that
  requests for communication.}). So, by this property we force the
communicating component to know when the communication occurs at its
state level. This is done by imparting a switching element inside the
state of the system that becomes {\em one} if the communication
occurs, else it will be {\em zero}. This property in $\RCPCPA$ is used
to get the result that a $\RCPCPA$ can be simulated by a multi-head
pushdown automata. This result, in turn, give rise to a new class in
$\RPCPA$ called as $\simpleRPCPA$ that can be equivalently written as
finite intersection of multi-head pushdown automata systems.

The paper is organized in the following way. Section~\ref{section2}
consists of some preliminaries -- definition of parallel communicating
pushdown automata systems and multi-head pushdown automata. In
Section~\ref{section3}, we present a property of PCPA called as known
communication and in the following section we prove that $\RCPCPA$ can
be simulated by a multi-head pushdown automata thereby showing that
the class of languages accepted by $\RCPCPA$ is contained in the class
of languages accepted by multi-head pushdown
automata. Section~\ref{section5} defines a restrictive class of
$\RPCPA$ called as $\simpleRPCPA$ and show that the language accepted
by it can be written as a finite intersection of languages accepted by
multi-head pushdown automata. The paper concludes with some remarks in
Section~\ref{section6}.

\section{Background}

\label{section2}
We assume that the reader is familiar with the basic concepts of
formal language and automata theory, particularly the notions of
grammars, grammar systems and pushdown automata. More details can be
found in~\cite{1,11,12}.

An alphabet is a finite set of symbols. The set of all words over
an alphabet~$V$ is denoted by $V^*$. The empty word is written as
$\epsilon$ and $V^+=V^*-\{\epsilon\}$. For a finite set $A$, we denote by
$\card{A}$ the cardinality of $A$; for a word $w\in V^*$,  $\card{w}_A$
denotes the number of symbols in $w$ which are from the set
$A\subseteq V$.  

We first define parallel communicating pushdown automata systems as
given in \cite{3}. 

\begin{definition}
\label{defn1}
A parallel communicating pushdown automata system of degree $n$ ($n\geq 1$)
is a construct $${\mathcal A}=(V, \Delta, A_1, A_2, \dots , A_n,
K)$$ where $V$ is the input alphabet, $\Delta$ is the alphabet of
pushdown symbols, for each $1\leq i\leq n$,
$$A_i=(Q_i,V,\Delta,f_i,q_i,Z_i,F_i)$$ is a pushdown automaton with the
set of states $Q_i$, the initial state $q_i\in Q_i,$ the alphabet of
input symbols~$V,$ the alphabet of pushdown symbols $\Delta,$ the
initial pushdown symbols $Z_i\in \Delta,$ the set of
final states $F_i\subseteq Q_i,$ transition mapping $f_i$
from $Q_i\times (V\cup\{\epsilon\})\times \Delta$ into the finite
subsets of ${Q_i\times \Delta^*}$ and $K \subseteq\{K_1,K_2,\dots
,K_n\}\subseteq \Delta$ is the set of query symbols. 

The automata $A_1,A_2,\dots,A_n$ are called the components of the
system~${\mathcal A}$. 
\end{definition}

If there exists only one $i$, $1\le i\le n$, such that for $A_i$, 
$(r,\alpha)\in f_i(q,a,A)$ with $\alpha\in \Delta^*, \ |\alpha|_K>0$
for some $r,q\in Q_i$, $a\in V\cup\{\epsilon\}$, $A\in \Delta$, then
the system is said  to be {\it centralized} and $A_i$ is said to be
the {\it master} of the system, i.\,e., only one of the component, called
the the master, is allowed to introduce queries. For the sake of
simplicity, whenever a system is centralized its master is taken to be
the first component unless otherwise mentioned. 

The configuration of a parallel communicating pushdown
automata system is defined as a $3n$-tuple
$(s_1,x_1,\alpha_1,s_2,x_2,\alpha_2,\dots ,s_n,x_n,\alpha_n)$
where for $1\leq i\leq n$, $s_i\in Q_i $ is the current state of the
component $A_i$, $x_i\in V^*$ is the remaining part of the input word
which has not yet been read by $A_i$, $\alpha_i\in \Delta^*$ is the
contents of the $i^{th}$ stack, its first letter being the topmost symbol.

The initial configuration of a parallel communicating pushdown
automata system is defined as\linebreak
$(q_1,x,Z_1,q_2,x,Z_2, \ldots, \ldots, q_n,x,Z_n)$
where $q_i$ is the initial state of the component $i$, $x$ is the
input word, and $Z_i$ is the initial stack symbol of the component
$i$, $1 \leq i \leq n$. It should be noted here that all the
components receive the same input word $x$. 

There are two variants of transition relations on the set of
configurations of~${\mathcal A}$. They are defined in the following
way:
\begin{enumerate}
\item $(s_1,x_1,B_1\alpha_1,\dots ,s_n,x_n,B_n\alpha_n)
\vdash (p_1,y_1,\beta_1,\dots ,p_n,y_n, \beta_n),$\\ where $B_i\in
\Delta, \alpha_i,\beta_i\in \Delta^*,$ $1\leq i\leq n,$ iff one of the
following two conditions hold:
\begin{enumerate}
\item[(i)] $K\cap \{B_1,B_2,\dots,B_n\}=\emptyset\mbox{ and
}x_i=a_iy_i,a_i\in V\cup\{\epsilon\}$,\\ $(p_i,\beta_i')\in
f_i(s_i,a_i,B_i),$ $\beta_i=\beta_i'\alpha_i,\ 1\le i\le n,$
\item[(ii)]  
(a) $\mbox{for all } i, 1\le i\le n \mbox{ such that }
B_i=K_{j_i}\mbox{ and } B_{j_i}\!\notin\! K, \
\beta_i=B_{j_i}\alpha_{j_i}\alpha_i,$  \\
(b) $\mbox{for all other } r, \ 1\le r\le n, \ \beta_r=B_r\alpha_r,
\mbox{ and }$ \\
(c) $y_t=x_t,\ p_t=s_t, \ \mbox{for all} \ t,\ 1\le t\le n.$
\end{enumerate}
\item $(s_1,x_1,B_1\alpha_1,\dots ,s_n,x_n,B_n\alpha_n)
\vdash_r (p_1,y_1,\beta_1,\dots ,p_n,y_n,
\beta_n),$\\
where $B_i\in \Delta, \alpha_i,\beta_i\in \Delta^*,$ $1\leq
i\leq n,$ iff one of the following two conditions hold:
\begin{enumerate}
\item[(i)] $K\cap \{B_1,B_2,\dots,B_n\}=\emptyset\mbox{ and
}x_i=a_iy_i,a_i\in V\cup\{\epsilon\}$,\\ 
$(p_i,\beta_i')\in
f_i(s_i,a_i,B_i), \beta_i=\beta_i'\alpha_i,\ 1\le i\le n,$ 
\item[(ii)] (a) $\mbox{for all } 1\le i\le n \mbox{ such that } 
B_i=K_{j_i}\mbox{ and } B_{j_i}\notin K$, \\  
${}\qquad\beta_i=B_{j_i}\alpha_{j_i}\alpha_i,\mbox{ and
}\beta_{j_i}=Z_{j_i},$\\
(b) $\mbox{for all the other } r, 1\le r\le n,
\beta_r=B_r\alpha_r, \mbox{ and }$\\
(c) $y_t=x_t,\ p_t=s_t, \mbox{for all} \ t,\ 1\le
t\le n$. 
\end{enumerate}
\end{enumerate}
The communication between the components has more priority than the
usual transitions in individual components. So, whenever a component
has a query symbol in the top of its stack it has to be satisfied by
the requested component before proceeding to the usual transitions.

The top of each communicated stack must be a non-query symbol before
the contents of the stack can be sent to another component. If the
topmost symbol of the queried stack is also a query symbol, then first
this query symbol must be replaced with the contents of the
corresponding stack. If a circular query appears, then the working of the
automata system is blocked.

After communication, the stack contents of the sender is retained in
the case of relation $\vdash$, whereas in the case of $\vdash_r$ it
looses all the symbols and the initial stack symbol is inserted into
the respective stack that communicated.

A parallel communicating pushdown automata system whose computations are
based on relation $\vdash$ is said to be {\it non-returning}; if its
computations are based on relation $\vdash_r$ it is said to be {\it
returning}.  

The language accepted by a parallel communicating pushdown automata
system, ${\mathcal A}$ is defined as
{\small\begin{gather*}
L({\mathcal A})=\{x\in V^*\!\mid(q_1,x,Z_1,\dots ,q_n,x,Z_n)\vdash^*
  (s_1,\epsilon,\alpha_1,\dots ,s_n,\epsilon,\alpha_n), s_i\in
  F_i,1\le i\le n\},\\
L_r({\mathcal A})=\{x\in V^*\!\mid(q_1,x,Z_1,\dots ,q_n,x,Z_n)\vdash_r^*
  (s_1,\epsilon,\alpha_1,\dots ,s_n,\epsilon,\alpha_n), s_i\in
  F_i,1\le i\le n\}
\end{gather*}}%
where $\vdash^*$ and $\vdash_r^*$ denote the reflexive and transitive
closure of $\vdash$ and $\vdash_r$ respectively.

We use these notations: $\RCPCPA(n)$ for returning
centralized parallel communicating pushdown automata systems of degree
at most $n$, $\RPCPA(n)$ for returning non-centralized parallel
communicating pushdown automata systems of degree at most $n$,
$CPCPA(n)$ for centralized parallel communicating pushdown automata
systems of degree at most $n$, and $\PCPA(n)$ for parallel communicating
pushdown automata systems of degree at most $n$. 

If $X(n)$ is a type of automata system, then ${\mathcal L}(X(n))$ is
the class of languages accepted by pushdown automata systems of type
$X(n)$. For example, ${\mathcal L}(\RCPCPA(n))$ is the class of languages
accepted by automata of the type $\RCPCPA(n)$ (returning centralized
parallel communicating pushdown automata systems of degree at most
$n$). Likewise, ${\mathcal L}(\RCPCPA)$ denotes the class of languages
accepted by automata of type $\RCPCPA(n)$ where $n$ is arbitrary.

\section{Known communication}
\label{section3}
In this section we define known communication property for parallel
communicating pushdown automata systems. Informally speaking, a PCPA
with known communication property is a PCPA wherein each component
that has communicated knows about whether or not it has communicated
its stack symbols in the previous step. We define the known
communication property more formally in the following.

\begin{definition}
\label{known-communication}
Let $\mathcal A$ be a parallel communicating pushdown automata system
of degree $n$ given by \hbox{${\mathcal A}=(V, \Delta, A_1, A_2, \dots ,
A_n, K)$} where for each $1\leq i\leq n$,
$A_i=(Q'_i,V,\Delta,f_i,q_i,Z_i,F_i)$. Let $Q'_i$ for each $i$ be
represented as $Q_i\times\{0,1\}$ where $Q_i$ is as in the definition
of a PCPA (Definition~\ref{defn1}). Then $\mathcal A$ is said to
follow known communication property if at a communication step
component~$j$ has communicated to component $i$ with the state of
those systems as $(q_j,0)$ and~$(q_i,0)$ then in the the next usual
transition of the component the state of the component $j$ will be of
the form $(q_j,1)$. And, it will not reach the state of the form
$(q_j,1)$ in any other case.
\end{definition}

The set $\{0,1\}$ acts as a switch for each component $j$, i.\,e., it
becomes $1$ when the component has communicated in the previous step,
or otherwise it is $0$.

First, we prove that every $\RCPCPA$ can be rewritten as a $\RCPCPA$ with
known communication property.

\begin{theorem}
\label{thm1}
For every $\RCPCPA$ ${\mathcal A}$ there exists an equivalent $\RCPCPA$
${\mathcal A'}$ that satisfies the known communication property.
\end{theorem}

\begin{proof} 
Let ${\mathcal A}=(V, \Delta, A_1, A_2, \dots ,
A_n, K)$ where for each $i$, $1\leq i\leq n$,
\hbox{$A_i=(Q_i,V,\Delta,f_i,q_i,Z_i,F_i)$}. We will construct an equivalent
$\RCPCPA$, ${\mathcal A'}$ that satisfies the known communication
property. We construct~${\mathcal A'}$ from ${\mathcal A}$ as
described below. Take ${\mathcal
  A'}=(V,\Delta',A'_1,A'_2,\ldots,A'_n,K)$ with
$$A'_i=(Q'_i,V,\Delta',f'_i,(q_i,0),Z_i,F'_i)$$ where
\begin{enumerate}
\item $\Delta'=\Delta \cup \{Z'_i\}$, $1\leq i \leq n$,
\item $Q'_i=Q^{(j)}_i \times \{0,1\} \cup \{(q_i,0),(q'^{(j)}_i,0)\}$
  with $j=\{1,2\}$ and $1\leq i \leq n$, and
\item $F'_i=F^{(j)}_i\times \{0,1\}$ with $j=\{1,2\}$ and $1\leq i
  \leq n$.
\end{enumerate}
The transition function $f'_i (1\leq i\leq n)$ is defined as follows:
\begin{enumerate}
\item $f'_i((q_i,0),\lambda,Z_i)=\{((q'^{(1)}_i,0),Z'_i)\}$,
\item $f'_i((q'^{(1)}_i,0),\lambda,Z'_i)=\{((q'^{(2)}_i,0),Z'_i)\}$,
\item $f'_i((q'^{(2)}_i,0),a,Z'_i)$ includes
  $\{((p^{(1)}_i,0),\alpha)\}$ where $(p_i,\alpha) \in
  f_i(q_i,a,Z_i)$ and $\alpha \in \Delta'^*$,
\item $f'_i((p^{(1)}_i,0),\lambda,X)=\{((p^{(2)}_i,0),X)\}$,
\item $f'_i((p^{(2)}_i,0),a,X)$ includes $\{((r^{(1)}_i,0),\alpha)\}$
  where $(r_i,\alpha) \in f_i(p_i,a,X)$ and $\alpha \in \Delta'^*$,
\item $f'_i((p_i^{(1)},0),\lambda,Z_i)=\{((p^{(2)}_i,1),Z'_i)\}$,
\item $f'_i((p^{(2)}_i,1),a,X)$ includes $\{((r^{(1)}_i,0),\alpha)\}$
  where $(r_i,\alpha) \in f_i(p_i,a,X)$ and $\alpha \in \Delta'^*$,
\item $f'_1((p^{(1)}_1,0),\lambda,Z'_j)=\{((p^{(2)}_1,0),\lambda)\}$
  where $j>1$.
\end{enumerate}
For the above, without loss of generality we assume that there are no
transitions of the form\linebreak $(q'_i,\lambda)\in f_i(q_i,a,Z_i)$. If there
were any such transitions we can replace it with $(q'_i,Z_i)\in
f_i(q_i,a,Z_i)$.

The main idea of the above construction is explained in the
sequel. Whenever the master component requests for communication from,
say, $j^{th}(1<j\leq n)$ component the $j^{th}$ component communicates
all the symbols from its stack and thereby it loses all its stack
contents. And so, it again starts its processing with the start stack
symbol $Z_j$ (since it is following the returning mode). Hence when the
communicating component suddenly sees a $Z_j$ in its stack, it might
be because of a communication step or it can even be a normal step
where simply by a sequence of erasing transitions it came to the start
stack symbol. Now to construct the new machine ${\mathcal A}'$ we
rewrite the transition in such a way that we put duplicate start stack
symbols $Z'_j$ for each component and make sure that the normal
transitions do not go beyond $Z'_j$ in the stack (this is done by
first two transitions). So, if this is done it will make sure that when
the $j^{th}$ component transition sees $Z_j$, it means that it has
encountered a communication step and hence it can store the
information in its state by making the switch of its state as {\em
  one}. When imparting duplicate start stack symbols $Z_j'$ for each
component $j$ and when this same symbol is communicated to the master
component together with other symbols the synchronization might get
affected. Hence to keep the synchronization intact, we break each
transition in ${\mathcal A}$ into two transitions in ${\mathcal A}'$
-- one, as a simple $\epsilon$ move and the other one as the usual
transition that simulates the transition of ${\mathcal A}$. So when
the master component encounters duplicate start stack symbols, it will
simply pop the symbol from the stack in the first $\epsilon$
transition and then continues with the usual transition from then
onwards. The breaking of each transition into two transition
facilitates to keep the synchronization intact. From the construction
it should be easy to observe that the switch of the state
corresponding to the communicating component becomes $1$ only when the
previous step was a communication step.

Hence it should be clear that both the machines ${\mathcal A}$ and
${\mathcal A}'$ accept the same language.
\end{proof}

Now, naturally, the following question arises: does the above
construction holds good for other variants of parallel communicating
pushdown automata systems, namely, non-centralized non-returning,
centralized non-returning and non-centralized returning. In case of
non-returning variants it will not hold because after the
communication takes place the respective communicating stacks retain
the symbols and so there is no way of checking if it has been
communicated or not. At this juncture we do not know if non-returning
versions have known communication property or not. But in the case of
non-centralized returning parallel communicating pushdown automata
systems it holds good. Hence, we have the following theorem:

\begin{theorem}
\label{thm2}
For every $\RPCPA$ ${\mathcal A}$, there exists an equivalent $\RPCPA$
${\mathcal A'}$ that satisfies the known communication property.
\end{theorem}

In the above construction the main thing to note here is, in the case
of $\RCPCPA$ it is intrinsically known that to which component the
communication has taken place whereas in the case of $\RPCPA$ it will
not be possible to track which component received the
communication -- the reason being that there may be two or more
components that can request for communication at the same time. 
 
\section{Power of RCPCPA}
\label{section4}

In this section, we use the known communication property to show that
the acceptance power of $\RCPCPA$ is equivalent to that of multi-head
pushdown automata.

\begin{theorem}
\label{thm3}
For every $\RCPCPA$ ${\mathcal A}$, there exists a multi-head pushdown
automata $M$ such that\linebreak $L({\mathcal A})=L(M)$.
\end{theorem}

\begin{proof} Let ${\mathcal A}=(V, \Delta, A_1, A_2, \dots ,
A_n, K)$ where for each $1\leq i\leq n$,
$$A_i=(Q_i,V,\Delta,f_i,q_i,Z_i,F_i).$$ Without loss of generality we
assume that ${\mathcal A}$ satisfies known communication property. We
will construct an equivalent $n$-head pushdown automata $M$. In the
following construction just for our convenience sake and to reduce the
clumsiness in our construction we denote the switch of the component
$j$ when in state~$r$ by $switch_j(r)$. Hence the switch is not
included in the state itself as in Theorem~\ref{thm1} where the states
were of the form $(r,0)$ or $(r,1)$. We also assume without loss of
generality that there is no transition of the form $(p,\lambda) \in
\delta(q,a,Z_i)$. Let $M=(n,Q,V',\Delta',f,q_0,Z_0,F)$ where
\begin{enumerate}
\item $n$ denotes the number of heads,
\item $Q=\{[p_1,p_2,\ldots,p_n,i] \mid p_j\in Q_j,1\leq i,j \leq n\}$,
\item $V'=V$,
\item $\Delta'=\Delta$,
\item $q_0=[p_1,p_2,\ldots,p_n,1]$,
\item $Z_0=Z_1$,
\item $F=\{[End_1,End_2,\ldots,End_n,n]\}$.
\end{enumerate}

Generally, in a $\RCPCPA$ system there will be two sets of components
-- one, that communicates to the master and the other one that do not
communicate, or rather that is not requested by the component for
communication. First, we will define transitions that takes care of the
former set. 

The transition function is defined as follows:
\begin{enumerate}
\item  $f([p_1,p_2,\ldots,p_n,1],a_1,a_2,\ldots,a_n,Z_1)$ includes\\
  $([q,p_2,\ldots,p_n,1],1,0,\ldots,0,\alpha)$ if $(q,\alpha) \in
  f_1(p_1,a_1,Z_1)$, 
\item $f([p,p_2,\ldots,p_n,1],a,a_2,\ldots,a_n,X)$ where $X \neq Q_j
  (1<j\leq n)$ includes\\ $([q,p_2,\ldots,p_n,1],1,0,\ldots,0,\alpha)$
  if $(q,\alpha) \in f_1(p,a,X)$,
\item $f([p,p_2,\ldots,p_n,1],a,a_2,\ldots,a_n,Q_j) =
  ([p,p_2,\ldots,p_n,0,0,\ldots,0,j],Z_j)$ where $1<j\leq n$,
\item
  $f([p,p_2,\ldots,p_j,\ldots,p_n,j],a_1,a_2,\ldots,a_j,\ldots,a_n,Z_j)$
  includes\\
  $([p_1,p_2,\ldots,r,\ldots,p_n,j],0,0,\ldots,j,\ldots,j,\alpha)$ if
  $(r,\alpha) \in f_j(p_j,a,Z_j)$ where $1<j\leq n$,
\item For $switch_j(r)=0$, $f([p_1,p_2,\ldots,p_n,j],a,X)$ includes\\
  $([p_1,p_2,\ldots,s,\ldots,p_n],0,0,\ldots,1,\ldots,0,\alpha)$ if
  $(s,\alpha) \in f_j(r,a,X)$ where $1<j\leq n$,
\item For $switch_j(r)=1$,
  $$f([p_1,p_2,\ldots,r,\ldots,p_n,j],a_1,a_2,\ldots,a,\ldots,a_n,X) =
  ([p_1,p_2,\ldots,r,\ldots,p_n,1],0,0,\ldots,0,X)$$ where $1<j\leq n$.  
\end{enumerate}
We note that there are $n$ components and we need to simulate all $n$
pushdowns with one pushdown. To perform this, we simulate each component
sequentially starting with the master component which is usually the
first component. The states are taken as a $n+1$-tuple
$[p_1,p_2,\ldots,p_n,i]$ where $i$ denotes that~$i^{th}$ component is
being simulated.  When the master component through some transition
introduces a query symbol $Q_j$, the system shifts the control to the
$j^{th}$ component and starts simulating it.

The simulation is done by carrying out the transitions of $j^{th}$
component within the transition of the multi-head pushdown automata
by using:
\begin{itemize}
\item the $j^{th}$ head of $M$;
\item the transition of $f_j$ in $f$;
\item the single stack available to simulate all the stacks of
  ${\mathcal A}$ but one at a time. 
\end{itemize}
When simulating the $j^{th}$ component if the system arrives at a state
where the switch becomes one, then the system stops the simulation and
shifts back to simulating the master component (first component) using $f_1$. 

This cycle of shifting of control from set of transition of master to
the set of transition of the component $j$ and back forth happens for
every query symbol $Q_j$ appearing in the stack when using the
transition of the master (first) component. 

Hence the above steps takes care of the communication from the $j^{th}$
component to the master, i.\,e., the first component.

Now $M$ has to take care of those transitions of the non-communicating
components and those communicating components that do transitions that
are not communicated to the master component. This is carried out by
the following set of transitions.
\begin{enumerate}
\item For $p\in F_1$, $$f([p,p_2,\ldots,p_n,1],\dollar,a_2,\ldots,a_n,X)$$
  includes $$([End_1,p_2,\ldots,p_n,2],0,0,\ldots,0,Z_2).$$
\item For $1<j\leq n$,
  $$f([End_1,\ldots,End_{j-1},p,p_{j+1},\ldots,p_n,j],\dollar,\ldots,\dollar,a,a_{j+1},\ldots,a_n,X)$$
  includes
  $$([End_1,\ldots,End_{j-1},q,p_{j+1},\ldots,p_n,j],0,\ldots,1,\ldots,0,\alpha)$$
  if $(q,\alpha)\in f_j(p,a,X)$.
\item For $p\in F_j$ and $1<j< n$,
  $$f([End_1,\ldots,End_{j-1},p,p_{j+1},\ldots,p_n,j],\dollar,\ldots,\dollar,a_{j+1},\ldots,a_n,X)$$
  includes
$$([End_1,\ldots,End_{j-1},End_j,p_{j+1},\ldots,p_n,j+1],0,\ldots,0,Z_{j+1}).$$
\item  For $p\in F_n$,
  $$f([End_1,\ldots,End_{n-1},p,n],\dollar,\ldots,\dollar,\lambda,X)$$
  includes
  $$([End_1,End_2,\ldots,End_n,n],0,\ldots,0,\alpha).$$
\end{enumerate}
The second set of transitions given above takes care of the
transitions of (1) those components not requested for communication by
the master and (2) those components that communicated to master but
has not yet read the full input string.  Suppose the first component
(master component) in ${\mathcal A}$ reaches a final state $p$ and
reads the end marker $\dollar $, then $M$ reaches the state wherein the
first ordinate is $End_1$. After it reaches this stage, $M$ switches
to the simulation of the second component of ${\mathcal A}$ until it
reaches the end marker $\dollar $. When reaching the end marker if the
state of the second component of ${\mathcal A}$ is a final state then
$M$ reaches the state wherein the second ordinate is $End_2$ and
switches to the simulation of the third component. If all $n$
components in ${\mathcal A}$ after reading the entire input string
arrives at a final state then $M$ arrives at the final state
$[End_1,End_2,\ldots,End_n]$. And hence the string accepted by
${\mathcal A}$ will be accepted by $M$. By construction it should also
be clear that $M$ does not accept strings that are not accepted 
by~${\mathcal A}$.%
\end{proof}

\section{A variant of RPCPA}
\label{section5}

In this section we define a restricted version of $\RPCPA$ and show the
implications of known communication property on it.

Since $\RPCPA$ is a non-centralized version, two or more components can
query other components. So, there are few scenarios possible:
\begin{enumerate}
\item There might be some components which neither queries nor is queried.
\item Some components might query other components (possibly at the same time)
  but it is not queried by any other component.
\item Some components might be queried (possibly by more than one
  component at the same time) but it might not query any
  other component.
\item Some components might be queried and might query other
  components. Possibly it can occur at the same step.
\end{enumerate}

Keeping these scenarios in hand we define a restrictive version of
parallel communicating pushdown automata systems in its returning
version in the sequel.

Let ${\mathcal A}=(V, \Delta, A_1, A_2, \dots , A_n, K)$ where for
each $1\leq i\leq n$, $$A_i=(Q_i,V,\Delta,f_i,q_i,Z_i,F_i)$$ be a
$\RPCPA$. 

We denote the set $\IN$ as the set containing all querying components
of ${\mathcal A}$ and $\OUT$ as the set consisting all queried
components of ${\mathcal A}$. Moreover, we denote the set $\IN(A_j)$
(for $1\leq j \leq n$) as the set of all components that $A_j$
queries. Similarly, the set $\OUT(A_j)$ (for $1\leq j \leq n$) denotes
the set of all components that queried $A_j$. Now we define
$\simpleRPCPA$.

\begin{definition}
A returning parallel communicating pushdown automata system~${\mathcal
  A}$ is said to be a {\em simple returning parallel communicating
  pushdown automata system} if it satisfies the following two
conditions:
\begin{enumerate}
\item For any $A_i$, $A_j$ ($1\leq i,j\leq n$ and $i \neq j$), $\IN(A_i)\cap
  \IN(A_j)=\emptyset$, and
\item If $A_j \in \IN(A_i)$ ($1\leq i,j\leq n$) then $A_i \not\in
  \IN(A_k)$ for any other $k$ where $1 \leq k \leq n$.
\end{enumerate}
\end{definition}

We note here that the above definition of $\simpleRPCPA$ restricts the
normal $\RPCPA$ by not allowing any two components to query the same
component and not allowing a component that queries to act as a
communicator to any other querying component. These two restrictions
make the communication scenarios, mentioned above, a kind of one-way
communication. A little insight into the definition of $\simpleRPCPA$
would exemplify the fact that it can be decomposed into a collection
of $\RCPCPAs$. This can be seen as follows:

Assume that ${\mathcal A}$ is a $\simpleRPCPA$. Now consider the set
$\IN$ of ${\mathcal A}$ -- that contains all querying components of
${\mathcal A}$. Let $\card{\IN}$ be $m$. Treating each of these $m$
components as masters together with the components that each of these
$m$ components queries we can have $m$ number of $\RCPCPAs$. The
components which are neither querying nor getting queried have to be
combined with any of these $\RCPCPAs$. Hence we have decomposed
$\simpleRPCPA$ into $m$ number of $\RCPCPAs$.

The above discussion gives the following Lemma:

\begin{lemma} 
\label{lem1}
Given a $\simpleRPCPA$ ${\mathcal A}$, $L({\mathcal A})$ can be
written as $$L({\mathcal A}) = L({\mathcal A}_1)\cap L({\mathcal A}_2) \cap
\cdots \cap L({\mathcal A}_m)$$ where for each $i(1 \leq i \leq m)$,
${\mathcal A}_i$ is a $\RCPCPA$.
\end{lemma}

By Theorem~\ref{thm2} and Lemma~\ref{lem1} we have the following theorem:

\begin{theorem}
\label{thm4}
For any $\simpleRPCPA$ ${\mathcal A}$, $L({\mathcal A})$ can be written
as $$L({\mathcal A})=L(M_1)\cap L(M_2)\cap \cdots \cap L(M_m)$$ where $m\geq 1$
and each $M_i (1 \leq i \leq m)$ is a multi-head pushdown automata.
\end{theorem}

\begin{proof*} Let 
${\mathcal A}=(V, \Delta, A_1, A_2, \dots, A_n, K)$ where for each $1\leq i\leq n$,
$$A_i=(Q_i,V,\Delta,f_i,q_i,Z_i,F_i)$$ be a $\simpleRPCPA$. By
Lemma~\ref{lem1}, $L({\mathcal A})$ can be written as finite intersection
of $\RCPCPA$. Hence $L({\mathcal A}) = L({\mathcal A}_1)\cap L({\mathcal A}_2)
\cap \cdots \cap L({\mathcal A}_m)$ where for each $i$ ($1 \leq i \leq m$),
${\mathcal A}_i$ is a $\RCPCPA$.

Now by Theorem~\ref{thm2}, for each $\RCPCPA$ there exists an equivalent
multi-head pushdown automata~$M$. Hence if $M_i$ is the corresponding
multi-head pushdown automata for ${\mathcal A}_i$ then we
have 
\begin{equation}
L({\mathcal A})=L(M_1)\cap L(M_2)\cap \cdots \cap L(M_m).\tag*{\qed}
\end{equation}
\end{proof*}
 
\section{Conclusion}
\label{section6}
We defined a property called known communication for parallel
communicating pushdown automata in general. Using this property we
showed that the acceptance power of the returning centralized pushdown
automata ($\RCPCPA$), which was left open in \cite{3}, is equivalent to
that of multi-head pushdown automata.  The above result also implies
that a restrictive class of returning parallel communicating pushdown
automata can be written as a finite intersection of multi-head
pushdown automata.

The known communication property for the non-returning variants of
PCPA is still open. Our proof for the returning variants of PCPA does
not hold good for the non-returning variants as the communicating
stacks retain the stack contents after communication and there is no
way of checking if they have been communicated or not. One possible
way of checking if the communication is taking place is by polling each
component by looking for any query symbols in its stacks. But this would
again amount to some more communications. And so this will result in a
recursive argument. Another interesting problem to look into is the
following. Can we relax the restrictions in $\simpleRPCPA$ and still
get the same result as in Theorem~\ref{thm4}?

\bibliographystyle{eptcs}
\bibliography{balan}

\end{document}